\documentclass{article}

\usepackage{amsthm}
\usepackage{amssymb}
\usepackage{amsfonts}
\usepackage{mathrsfs,epsfig}
\usepackage{systeme, mathtools}
\usepackage{tikz-cd}
\usepackage{float,bbm}

\usepackage{authblk}

\usepackage{xcolor}
\usepackage{graphicx}
\usepackage{dcolumn}
\usepackage{bm}
\usepackage{appendix}
\usepackage{url}
\usepackage{hyperref}

\hypersetup{hidelinks}

\numberwithin{equation}{section}
\newtheorem{theorem}{Theorem}[section]

\newtheorem{proposition}[theorem]{Proposition}

\theoremstyle{definition}

\newtheorem{definition}[theorem]{Definition}

\newcommand{\R}{\mathbb{R}}

\newcommand{\g}{\mathfrak{g}}

\newcommand{\so}[1]{\mathfrak{so}(#1)}
\newcommand{\M}{\mathcal{M}}
\newcommand{\T}{\mathrm{T}}

\newcommand{\Ad}{\operatorname{Ad}}

\newcommand{\inn}{\operatorname{\lrcorner}}

\newcommand{\sslash}{\mathbin{/\mkern-6mu/}}
\newcommand{\pder}[2]{\frac{\partial #1}{\partial #2}}

\title{Generalized Uncertainty Principle theory with a single constraint}
\author[1,2]{Matteo Bruno\footnote{\texttt{matteo.bruno@uniroma1.it }}}
\author[1]{Sebastiano Segreto\footnote{\texttt{sebastiano.segreto@uniroma1.it }}}

\affil[1]{\textit{\normalsize{Physics Department, Sapienza University of Rome, P.za Aldo Moro 5, Rome, 00185, , Italy}}}
\affil[2]{\textit{\normalsize{INFN, Sezione di Roma 1, P.le Aldo Moro 2, 00185, Rome, Italy}}}
\date{}

\begin{document}

\maketitle

\begin{abstract}
We aim to analyze the consistency of the deformation of the Heisenberg algebra in the setting of constrained Hamiltonian systems, providing a procedure to induce the deformation on the Poisson algebra after symplectic reduction. We investigate this in the context of the classical interpretation of Generalized Uncertainty Principle theories, treating two cases separately. For the first case, we consider a group action on the phase space together with a set of first-class constraints that can be interpreted as a momentum map. We furnish an explicit example in the case of rotational invariant deformed algebras. In the second case, we consider a single constraint provided by the Hamiltonian, which is a common instance in General Relativity, with straightforward application in cosmology.
\end{abstract}

\section{Introduction}

Generalized Uncertainty Principle (GUP) theories belong to the category of effective theories that provide an alternative description of the structure of space (or of configuration space, more generally), exhibiting features compatible with the quantum structure suggested by a plethora of quantum gravity theories \cite{konishi1990minimum, maggiore1993generalized}. In this sense, they define a framework that can be readily employed to test the consequences of presumed quantum gravity effects on the dynamics of the system of interest \cite{ali2009discreteness, das2009phenomenological}. Their quantum formulation is based on the deformation of the ordinary Heisenberg algebra between quantum conjugate operators \cite{Kempf:1993bq, Kempf:1994su}. Indeed, this modification is the seed for introducing nonlocal effects into the theory, resulting — according to the specific algebra — in the emergence of a minimal length and the appearance of non-commutativity between configuration operators \cite{segreto2023extended, Segreto:2024vtu}. In other words, the foundation of these theories is the proper deformation of the standard canonical structure of quantum mechanics, with all the consequences such a modification implies. \\
The nature of the relation between the GUP algebra and the alteration of the canonical structure on which standard physics is based is more evident in a classical formulation of the framework.\\
Indeed, the deformed quantum commutators can be inherited — at least formally — by a set of Poisson brackets (PBs) describing a deformed classical mechanics. The structure of these PBs is completely described by a 2-form $\omega$, where all the information of this new classical theory is encoded \cite{Lee_2014}. As we were able to show extensively and rigorously in one of our previous works \cite{Bruno_Segreto_2024}, this 2-form — under a precise set of conditions regarding the functions that control the deformation of the PBs — is a symplectic form that arises as a deformation of the ordinary one. This fact is exactly what enables a classical Hamiltonian formulation of GUP theories, an aspect that is extremely relevant given the semi-classical interpretation these frameworks naturally carry to study quantum gravity effects in a transition regime (see e.g. \cite{segreto2025dynamics} for the primordial cosmological scenario).\\
This paper seeks to extend the analysis developed in~\cite{Bruno_Segreto_2024} by incorporating well-established constructions in symplectic geometry, such as induced structures on submanifolds and the symplectic quotient. Broadly speaking, these operations are required whenever one studies a physical system that exhibits gauge symmetries, and therefore is subject to constraints. Such situations are ubiquitous in physics—most notably in the context of gauge fields and gravity—and thus provide highly general and physically relevant scenarios. The aim of the present work is to offer a systematic prescription for treating GUP theories, in their classical formulation, when constraints of this type come into play.
\smallskip
In the first part of the paper (Sec.~\ref{Sec_II}), we address the constraints by employing the general framework of symplectic reduction~\cite{marsden1974reduction}. In this setting, the phase space is assumed to carry an action of a Lie group, interpreted as the gauge symmetry group. We consider the cases in which the gauge group is $SO(2)$ or $SO(3)$, and subsequently restrict our analysis to rotationally invariant GUP-deformed algebras~\cite{Bruno_Segreto_2024, Maggiore:1993kv}.
Identifying the momentum map $\mu$ with the generator of rotations, we explicitly compute the symplectic structure, in both cases, on the reduced phase space $\mathcal{M} \sslash SO(n)$. We further characterize its topology by observing that the constraint surface $\mu^{-1}(0)$ possesses an invariant vector bundle structure. As expected, the symplectic reduction procedure preserves the geometric properties of the original phase space while eliminating redundant degrees of freedom, as evidenced by the form of the reduced symplectic structure, which retains the deformed features of the underlying model.
\smallskip
The second part of the paper (Sec.~\ref{Sec_III}) is devoted to the case in which one deals with a single constraint, represented by the Hamiltonian itself~\cite{Henneaux:1992ig}. In this situation, no Lie-group action is available on which the standard symplectic reduction procedure can be constructed. To treat this case, we introduce a general prescription that is closer in spirit to well-known gauge-fixing techniques: we define an ad hoc vector field $T$ on the phase space, which plays the role of an external time parameter and is required to properly define a dynamics on a suitable submanifold of the full phase space. We identify a condition on both the symplectic structure and the choice of $T$ that guarantees the internal consistency of the proposed method.\\
This second scenario corresponds precisely to Hamiltonian cosmology, where the constrained nature of General Relativity is encoded in the Hamiltonian constraint $H = 0$~\cite{Arnowitt:1962hi, Misner:1969ae}. Given the relevance of the reduction procedure in this physical setting, we explicitly analyze the example provided by the Hamiltonian formulation of Bianchi models~\cite{ellis2006bianchi}, using the Misner variables~\cite{Misner:1969ae}. Following the general prescription introduced above, we are able to determine explicitly the reduced phase space, the reduced Hamiltonian, and the reduced symplectic form. Once again, the resulting symplectic structure displays the same functional form as the original one, thereby preserving the geometry of the phase space throughout the reduction procedure.\\
Nevertheless, as we show explicitly, the success of this construction requires certain conditions to be imposed on the symplectic form of the original phase space. In particular, with respect to the Poisson-bracket structure, no noncommutativity may occur between the variable denoted by $q_0$ and the remaining variables $q_i$. Since, in the reduction procedure, $q_0$ is ultimately selected as the clock variable governing the dynamics on the submanifold, this indicates that time variables cannot be noncommutative with spatial variables—at least within the framework of our construction. This observation may be related to the loss of unitarity observed in quantum theories that allow such forms of noncommutativity~\cite{balachandran2004unitary}.\\\
Finally, the procedure we outline provides a solid prescription to correctly treat the reduction procedure in cosmology in this deformed sector. Since the structure of the reduced symplectic form is the same as that of the original one, only with respect to the newly selected degrees of freedom, the present study represents a complete validation of the naiver approach according to which the deformed symplectic form is directly imposed on the reduced phase space, showing thus complete compatibility of the dynamical results obtained through the two different methods.

\section{Deformed algebra after symplectic reduction} \label{Sec_II}
The class of GUP theories of which we aim to discuss the implementations of the constraints in its classical interpretation is the one described by the following commutation relations:
  \begin{align}
  \label{GUP}
    \nonumber
    &[\hat p_i,\hat p_j]=0,\\
    &[\hat q_i,\hat q_j]=i\hbar L_{ij}(\hat q,\hat p),\\ \nonumber    
    &[\hat q_i,\hat p_j]=i\hbar\delta_{ij} f(\hat q,\hat p).
  \end{align}
  
 We recall that by \emph{classical interpretation} of a \emph{d}-dimensional GUP theory we intend the following \cite{Bruno_Segreto_2024}:
 \begin{definition}
 \label{class}
    The classical interpretation of a GUP theory in the class given above consists of interpreting the set $(q_1,\dots,q_d,p_1,\dots,p_d)$ as coordinates on a \emph{2d}-dimensional smooth manifold $\M$ equipped with a symplectic form $\omega$ such that it induces the following fundamental Poisson brackets
    \begin{align}
    \label{Poisson}
    \nonumber
    &\{ p_i, p_j\}=0,\\
    &\{ q_i, q_j\}= L_{ij}(q,p),\\ \nonumber  
    &\{ q_i,p_j\}=f(q,p)\delta_{ij}.
    \end{align}
    
    We refer to the symplectic manifold $(\M,\omega)$ as phase space.
 \end{definition}
  From Definition \ref{class}, we can compute the matrix associated with the symplectic form that induces the Poisson brackets \eqref{Poisson} in the coordinate system $x^a=(q_1,\dots,q_d,p_1,\dots,p_d)$. Indeed, the inverse of the symplectic form matrix is easily computable:
 \begin{equation}
 \label{PoissonMat}
   \omega^{ab}=\{x^a,x^b\}=
   \begin{pmatrix}
    L & f\mathrm{id}_d\\
    -f\mathrm{id}_d & 0
   \end{pmatrix},
 \end{equation}
 The inverse of this matrix is immediate and so we obtain the matrix of the symplectic form
 \begin{equation}
 \label{Sp-form}
  \omega_{ab}=
  \begin{pmatrix}
    0 & -\frac{1}{f}\mathrm{id}_d\\
    \frac{1}{f}\mathrm{id}_d & \frac{1}{f^2}L
  \end{pmatrix},
 \end{equation}
 where clearly $L:=\{L_{ij}\}$ is a skew-symmetric submatrix and $\mathrm{id}_d$ is the identity matrix in $d$ dimensions.
 
\smallskip

Given a symplectic action of a Lie group $G$ on a symplectic manifold $\M$ with symplectic form $\omega$, it may exist a momentum map for this action. Namely, a map $\mu:\M\to \g^*$ such that
\begin{enumerate}
    \item $d\mu(\xi)=\varrho(\xi)\inn\omega$,
    \item $\mu\circ g=\Ad^*_g\mu$.
\end{enumerate}
If it exists, we say that the action is Hamiltonian. Here, $\varrho$ is the Lie algebra homomorphism $\varrho:\g\to\mathfrak{X}(\M)$ induced by the action.\\
If $G$ acts freely and properly on $\mu^{-1}(0)$, then $0$ is a regular value for $\mu$ and the quotient $\mu^{-1}(0)/G$ is a smooth manifold. The quotient inherits a symplectic form from $\M$; that is, there is a unique symplectic form on the quotient space whose pullback to $\mu ^{-1}(0)$ equals the restriction of $\omega$ to $\mu ^{-1}(0)$. Thus, the quotient is a symplectic manifold, called the symplectic reduction of $\M$ by $G$ and is denoted $\M\sslash G$ \cite{Marsden_Weinstein_1974}.

\medskip

In classical interpretation of Generalized Uncertainty Principle theories the deformation of the algebra is encoded in interpreting the physical position $q$ and momentum $p$ as coordinate on a symplectic manifold that are not a Darboux frame \cite{Bruno_Segreto_2024}. However, since a set of Darboux coordinate always exists on a symplectic manifold, we need to characterize the coordinate system in which to evaluate the symplectic form on the symplectic reduction starting from the initial coordinate frame.
\begin{definition}
    A coordinate system $(x_1,\dots,x_m)$ on $\M\sslash G$ is called natural with respect to the classical interpretation of a GUP theory if it is possible to induce coordinates on $\mu^{-1}(0)$ of the form $(\theta_1,\dots,\theta_l)$ such that $(x_1,\dots,x_m)=\pi(\theta_1,\dots,\theta_l)$, where $\pi$ is the projection $\mu^{-1}(0)\to\mu^{-1}(0)/G=\M\sslash G$, and providing the parametric description $\left(q_1(\theta),\dots,q_n(\theta),p_1(\theta),\dots,p_n(\theta)\right)$ of $\mu^{-1}(0)$ in $\M$, with the additional condition that, if $(q_i,p_i)$ forms a Darboux frame when $f(p)=1$ and $L_{ij}=0$, then $(x_1,\dots,x_m)$ will be a Darboux frame for the induced symplectic form.
\end{definition}

\subsection{Rotational-invariant deformed algebra in two dimensions}
\label{Sec:2-d}
We can start from a simple toy model: we consider a two dimensional system with a rotational symmetry and a constraint given by the angular momentum. 
To satisfy these requests, we need to deal with GUP Poisson brackets with the following structure:
\begin{equation} \label{maggiore_algebra}
    \begin{split}
        &\{p_i,p_j\}=0, \\
    &\{q_i,q_j\}= a(p)J,\\  &\{q_i,p_j\}= f(p)\delta_{ij},
    \end{split}
\end{equation}
where $J$ is the generator of the rotations, i.e. $\{J,p_i\}=\epsilon_{ij}p_j$ and $\{J,q_i\}=\epsilon_{ij}q_j$, and the equation $J=0$ set the constraint.\\
This is exactly a two-dimensional version of the so-called Maggiore algebra with constraint $J=0$ \cite{Maggiore:1993kv}. \\
Explicitly we can write $J$ as follows \cite{Bruno_Segreto_2024}:
\begin{equation}
    J=\frac{1}{f(p)}(q_1p_2-q_2p_1).
\end{equation}
It is known that, in order the algebra \eqref{maggiore_algebra} to be rotational invariant, $a$ and $f$ must depend on the modulus of $p$ only. To carry out our computations, we consider as phase space $\M=\T^*(\R^2\setminus\{0\})\cong (\R^2\setminus\{0\})\times \R^2$ with a symplectic form in the coordinate system $(q_1,q_2,p_1,p_2)$ given by
\begin{equation}
    \omega_{ab}=\frac{1}{f}\begin{pmatrix}
        0 & 0 & -1 & 0\\
        0 & 0 & 0 & -1\\
        1 & 0 &  0 & \frac{a}{f}J \\
        0 & 1 & -\frac{a}{f}J & 0
    \end{pmatrix},
\end{equation}
where the functions $a$ and $f$ are linked by the following differential equation \cite{Bruno_Segreto_2024}:
\begin{equation}
    f\frac{\partial f}{\partial p_i}=-ap_i.
\end{equation}
For the sake of regularity we have removed the point $0\in \R^2$ from the configuration space. Indeed, intuitively, we have in mind a rotating particle, which would not be affected by rotation if placed at the origin of the configuration space.\\
Now, the action of $SO(2)$ on this space is the usual one, and it is represented by the matrix
\begin{equation*}
    \begin{pmatrix}
        \cos\theta & \sin\theta & 0 & 0\\
        -\sin\theta & \cos\theta & 0 & 0\\
        0 & 0 & \cos\theta & \sin\theta\\
        0 & 0 & -\sin\theta & \cos \theta
    \end{pmatrix}
\end{equation*}
The generator of rotations is the candidate to be the momentum map for this action: the coadjoint action is trivial because $SO(2)$ is Abelian, indeed $J$ is invariant under rotations. We can check that this action is Hamiltonian. \\ 
The Lie algebra $\so{2}$ is a real one-dimensional Abelian Lie algebra, hence, it is isomorphic to $\R$. The generator of this algebra is the matrix $\begin{pmatrix}
        0 & 1 \\
        -1 & 0 
    \end{pmatrix}$ and the corresponding vector field is
\begin{equation}
    \varrho\left(\begin{pmatrix}
        0 & 1 \\
        -1 & 0 
    \end{pmatrix}\right)=q_2\pder{}{q_1}-q_1\pder{}{q_2}+p_2\pder{}{p_1}-p_1\pder{}{p_2}.\
\end{equation}
From a simple computation in coordinate, we can check that $\xi\inn\omega$ is equal to
\begin{equation}
    dJ=\frac{1}{f}\left(p_2dp_1-p_1dp_2-(q_2+\frac{a}{f^2}p_1J)dp_1+(q_1+\frac{a}{f^2}p_2J)dp_2\right).
\end{equation}
We now study the orbit on the constraint $J=0$. Since $SO(2)$ is compact, its action on $\M$ is proper, moreover, the action is free since $0\in\R^4$ is not included. The map $J:\M\to \R$ has constant rank equal to $1$, then is a smooth submersion. Hence, the level set $J=0$ is a properly embedded submanifold of codimension 1 in $\M$. Furthermore, the condition $J=0$ can be written equivalently as $q\times p=0$ forcing $q$ and $p$ to be colinear. This means that the constraint surface $S=\{(q,p)\in\M\ |\ J(q,p)=0\}$ is a vector bundle of rank 1 with base manifold $\R^2\setminus\{0\}$. In fact, $S$ can be equivalently casted as:
\begin{equation}
    S=\{(q,p)\in(\R^2\setminus\{0\})\times \R^2\ |\ p\in Span(q)\subset\R^2\},
\end{equation}
clearly providing a submanifold of dimension 3. Moreover, this vector bundle is trivial since, for a fixed $\alpha \in\R\setminus\{0\}$, the map $q\mapsto(q,\alpha q)$ provides a everywhere non-null sections \cite{Kobayashi_Nomizu_1963}.\\
To better understand the global behavior, let us focus on the vector bundle structure. The constraint surface $S$ inherits the $SO(2)$ action from the full space. This action sends fiber into fiber, making $S$ a $SO(2)$-invariant vector bundle (see Appendix \ref{App:A} for a brief introduction to invariant vector bundles or \cite{Kobayashi_Nomizu_1969} for a exhaustive discussion). \\
Considering the trivialization of $S$ induced by the section
\begin{equation}
\begin{matrix}
    \sigma:&\R^2\setminus\{0\}&\to& S;\\
    & q & \mapsto & (q,q),
\end{matrix} 
\end{equation}
it provides an isomorphism with a $SO(2)$-invariant vector bundle $(\R^2\setminus\{0\})\times\R$ together with the action $(q,\alpha)\mapsto(R(\theta)q,\alpha)$, where $R(\theta)$ is the rotational matrix in the defining representation of $SO(2)$. From this we can conclude that the quotient space reads simply as $((\R^2\setminus\{0\})\times\R)/SO(2)\cong \R_+\times \R $, thus the topology of the symplectic reduction must be:
\begin{equation}
    \M\sslash SO(2)\cong \R_+\times \R.
\end{equation}
This suggests that it is possible to cover the whole manifold with a global chart.

\medskip

We parametrize the constraint surface considering $r>0,\,\rho\in\R$, and $\theta \in [0,2\pi)$ with
\[q_1=r\cos\theta,\,q_2=r\sin\theta,\,p_1=\rho\cos\theta,\,p_2=\rho\sin\theta,\]
so the projection $\pi:J^{-1}(0)\to\M\sslash SO(2)$ is simply $\pi(r,\rho,\theta)=(r,\rho)$, with $(r,\rho)$ providing a global coordinate system for $\M\sslash SO(2)$. Then the symplectic structure on the quotient reads as $\tilde\omega=w(r,\rho)dr\wedge d\rho$, from which its pullback by the projection $\pi$ is $\pi^*\tilde\omega=w(r,\rho)dr\wedge d\rho$. To find the function $w(r,\rho)$, we must impose $\omega|_{J^{-1}(0)}=\pi^*\tilde\omega$, where the induced 2-form can be computed using the parametrization of the constraint surface
\begin{equation}
\begin{split}
    \omega|_{J^{-1}(0)}&=\left(-\frac{1}{f}\delta_{ij}dq^i\wedge dp^j+\frac{a}{f}Jdp^1\wedge dp^2\right)|_{J^{-1}(0)}\\
    &=-\frac{1}{f(\rho)}\left(d(r\cos\theta)\wedge d(\rho\cos\theta)+d(r\sin\theta)\wedge d(\rho\sin\theta)\right)\\
    &=-\frac{1}{f(\rho)}dr\wedge d\rho.
\end{split}
\end{equation}
Hence, the symplectic form on $\M\sslash SO(2)$ is $\tilde\omega=-\frac{1}{f(\rho)}dr\wedge d\rho$. \\
It is important to notice that, once would be tempted to say that $r=\sqrt{q_1^2+q_2^2}$ and $\rho=\sqrt{p_1^2+p_2^2}$. However, this is not true. In that case, the coordinates $r,\rho$ would be just local coordinates for $\M\sslash SO(2)$. Indeed, they correctly cover only the open subset of the symplectic reduction given by the image of the points on $S$ of the form $p=\alpha q$ for $\alpha>0$ (or, equivalently, $\alpha<0$). Indeed, despite they are constant along the orbits, these coordinates cannot distinguish the orbit of the point $(q,\alpha q)$ to the one of $(q,-\alpha q)$. Moreover, in the case $\alpha=0$, $\rho$ fails to be smooth. A possible extension to overcome this problem is to define coordinates on $S$ with $q_1\neq 0$ as
\begin{equation}
    r=\sqrt{q_1^2+q_2^2},\ \ \ \rho=\mathrm{sgn}\left(\tfrac{p_1}{q_1}\right)\sqrt{p_1^2+p_2^2}.
\end{equation}
Clearly, together with the chart defined for $S\cap\{q_2\neq 0\}$
\begin{equation}
    r=\sqrt{q_1^2+q_2^2},\ \ \ \rho=\mathrm{sgn}\left(\tfrac{p_2}{q_2}\right)\sqrt{p_1^2+p_2^2}.
\end{equation}
the two cover the whole symplectic quotient. They agree when both $q_1$ and $q_2$ are not vanishing, and the form of the symplectic form remains the same. The fact that these maps are smooth is not immediate. We need to check what happens when $p_1$ or $\rho$ vanishes: first of all we notice that, when $q_1\neq 0$, if $p_1=0$ then $p_2=0$. Hence, we need to take care only of the vanishing of $\rho$: since $p$ and $q$ are colinear we can write $p_1=\alpha q_1$ and $p_2=\alpha q_2$, obtaining
\[\rho=\mathrm{sgn}\left(\tfrac{p_1}{q_1}\right)\sqrt{(\alpha q_1)^2+(\alpha q_2)^2}=\mathrm{sgn}\left(\tfrac{p_1}{q_1}\right)|\alpha|\sqrt{q_1^2+q_2^2}=\frac{p_1}{q_1}r,\]
ensuring the smoothness of the first map as projection from $S$ to $\M\sslash SO(2)$. For the second, we can follow an analogous procedure, recalling $\alpha=\frac{p_2}{q_2}$ when $q_2=0$. On the common domain, the change of coordinates is smooth since it is just the identity.

\bigskip

This procedure demonstrates how the deformation of the algebra is thus inherited by the symplectic quotient by choosing a coordinate frame induced by the starting one in a natural sense.

\subsection{Rotational-invariant deformed algebra in three dimensions}
\label{Sec:3-d}
We can provide an entirely analogous analysis for the three dimensional case. \\
In this case, the physical system admits an interpretation as a cosmological $SO(3)$ Yang-Mills theory in two-dimension. Here, the cosmological hypothesis of homogeneity is needed to reduced the system to finite-dimensional degrees of freedom, while the $SO(3)$ gauge symmetry provides a Gauss constraint that is reduced to the angular momentum expression within this context. \\
All the details of this specific physical model are worked out in Appendix~\ref{App:B}, while in what follows we provide the general prescription for the selected scenario, that is rotational-invariant deformed algebra in three dimensions. \\

We chose as phase space of the theory $\M=\T^*(\R^3\setminus\{0\})\cong(\R^3\setminus\{0\})\times\R^3$ and we impose the following GUP Poisson Brackets:
\begin{equation}
\label{Maggiore-GUP}
    \begin{split}
    &\{p_i,p_j\}=0, \\
    &\{q_i,q_j\}= a(p)\epsilon_{ijk}J_{k}, \\  &\{q_i,p_j\}= f(p)\delta_{ij},
    \end{split}
\end{equation}
where $J_i=\tfrac{1}{f}\epsilon_{ijk}q_jp_k$.
In this case, an element $g\in SO(3)$ acts via a rotational matrix $R_g$ in the defining representation as $(q,p)\mapsto(R_gq,R_gp)$. The basis of the Lie algebra $\so{3}$ is given by the matrices $T_i$ and we call the dual basis $\mathfrak{t}^i$. The momentum map for this action is given by:
\begin{equation}
    \begin{matrix}
        \mu: & \M& \to &\so{3}^*\\
        & (q,p)& \mapsto & J_i\mathfrak{t}^i.
    \end{matrix}
\end{equation}
By taking into account that $a(p)$ and $f(p)$ must depend only on the modulus of $p$, the $SO(3)$-equivariant of the momentum map follows:
\begin{align*}
    J_i(R_gq,R_gp)\mathfrak{t}^i=\frac{1}{f(Rp)}\epsilon_{ijk}R_{jl}R_{km}q_lp_m\mathfrak{t}^i=\frac{1}{f(p)}R_{il}\epsilon_{ljk}q_jp_k\mathfrak{t}^i=J_i(q,p)\Ad^*_g\mathfrak{t}^i.
\end{align*}
Considering an element $\xi^iT_i$ in the Lie algebra $\so{3}$, the corresponding vector field on $\M$ reads:
\begin{equation}
    \varrho(\xi)=\epsilon_{ijk}\xi^iq_k\pder{}{q_j}+\epsilon_{ijk}\xi^ip_k\pder{}{p_j}.
\end{equation}
Notice that $\xi^i$ are constants. \\
A simple computation of linear algebra shows that $\xi\inn\omega$ is equal to:
\begin{equation}
    \xi^idJ_i=\frac{1}{f}\left(\epsilon_{ijk}\xi^iq_jdp_k-\epsilon_{ijk}\xi^ip_kdq_j+\frac{a}{f}\xi^iJ_ip_jdp_j\right).
\end{equation}
The constraint surface $\mu^{-1}(0)$ is characterized by three equations $J_i=0$. The group action is proper because $SO(3)$ is compact, but is not free since whenever $q$ and $p$ are colinear, the point has a stabilizer isomorphic to $SO(2)$.  Furthermore, $0$ is not a regular value for this map. Indeed, in coordinate, the tangent map $d_p\mu:\T_p\M\to \T_0\so{3}^*\cong\R^3$ for $p\in\mu^{-1}(0)$ is given by:
\begin{equation}
    \begin{pmatrix}
        0 & -p_3 & p_2 & 0 & -q_3 & q_2\\
        p_3 & 0 & -p_1 & q_3 & 0 & -q_1\\
        -p_2 & p_1 & 0 & -q_2 & q_1 & 0
    \end{pmatrix}.
\end{equation}
By the \emph{bordered minors theorem}, it follows that ${\rm rk}(d\mu)=2$ on $\mu^{-1}(0)$. However, the rank is constant on $\mu^{-1}(0)$ as well as the dimension of the stabilizer group of each point. From the consideration on the rank, we can conclude that the constraint surface $\mu^{-1}(0)$ is locally an embedded submanifold of dimension $4$. In fact, the surface constraint is a smooth manifold of dimension $4$ that is a vector bundle over $\R^3\setminus\{0\}$ with fiber $\R$. As in the 2-dimensional case, the condition $\mu=0$ can be written equivalently as $q\times p=0$ forcing $q$ and $p$ to be colinear. Hence,
\begin{equation}
    S=\{(q,p)\in(\R^3\setminus\{0\})\times\R^3\ |\ p\in Span(q)\}\cong \mu^{-1}(0)
\end{equation}
is a trivial vector bundle $S\cong(\R^3\setminus\{0\})\times\R$. With similar considerations as in the two-dimensional case, we can deduce the topology of the symplectic quotient $\M\sslash SO(3)$. The constraint surface is a $SO(3)$-invariant vector bundle. The section
\begin{equation}
\begin{matrix}
    \sigma:&\R^3\setminus\{0\}&\to& S\\
    & q & \mapsto & (q,q),
\end{matrix} 
\end{equation}
is an invariant section because the stabilizer of $q$ is the same of $\sigma(q)$ as subgroup of $SO(3)$. Hence, the trivialization induced by $\sigma$ is a $SO(3)$-invariant vector bundle $(\R^3\setminus\{0\})\times\R$ with action of $SO(3)$ given by $R_g\times {\rm id}$ (cf. Appendix \ref{App:A}). The quotient of $\R^3$ by the defining representation of $SO(3)$ is known and we obtain
\begin{equation}
    \M\sslash SO(3)\cong ((\R^3\setminus\{0\})\times\R)/SO(3)\cong \R_+\times \R.
\end{equation}
Furthermore, every point in $\mu^{-1}(0)$ has as stabilizer a subgroup of $SO(3)$ conjugated to $SO(2)$. This condition is enough to ensure that $\mu^{-1}(0)$ is a manifold and the quotient $\mu^{-1}(0)/SO(3)$ has a natural symplectic structure. By Theorem~2.1 in \cite{Sjamaar_Lerman_1991}, the reduced space $\M\sslash SO(3)$ is constituted by a unique stratum which is naturally a symplectic manifold. A possible parametrization of the constraint surface is given by
\begin{equation}
    \begin{matrix}
        q_1=r\sin\theta \cos\phi, & p_1=\rho\sin\theta \cos\phi,\\
        q_2=r\sin\theta \sin\phi, & p_2=\rho\sin\theta \sin\phi,\\
        q_3=r\cos\theta, & p_3=\rho\cos\theta,
    \end{matrix}
\end{equation}
where $r>0,\,\rho\in\R$ and $\theta\in[0,\pi],\,\phi\in[0,2\pi)$. The induced form is $\omega|_{\mu^{-1}(0)}=\frac{1}{f(\rho)}d\rho\wedge dr$, since the projection is $\pi(r,\rho,\theta,\phi)=(r,\rho)$, the symplectic structure on $\M\sslash SO(3)$ reads
\begin{equation}
    \tilde\omega =\frac{1}{f(\rho)}d\rho\wedge dr.
\end{equation}

The procedure outlined demonstrates how even in this physically relevant case the deformation of the symplectic form is naturally preserved under symplectic reduction, thereby ensuring a well-posed application of the GUP framework for constrained Hamiltonian systems.

\section{Hamiltonian constraint in GUP} \label{Sec_III}
In this section we will deal with a different that must be treated differently, namely the case in which the theory has a single constraint $H=0$ . 
In this setting, we do not have the action of a group. Indeed, even if is intuitive to think the dynamics evolution induced by the Hamiltonian vector field $X_H$ as an action of $\R$, this is not correct since in general the vector field is not complete and so the flux cannot be defined for every time at every point. Thus, the technique of symplectic reduction can not be applied in this scenario. To treat this problem, we propose an approach closer to gauge fixing and relational observables.

\medskip

Let us consider a constrained Hamiltonian system $(\M,\omega,H)$, where $\M$ is a $(2n+2)$-dimensional symplectic manifold, $\omega$ is a symplectic structure, which in our preferred set of coordinates $(q_1,\dots,q_{n+1},p_1,\dots,p_{n+1})$ reads as in \eqref{Sp-form}, and $H$ is the Hamiltonian function and the constraint $H=0$ as well. On the level set $S=\{x\in\M\ |\ H(x)=0\}$, the Hamiltonian vector field is tangent, since $X_H\in\mathfrak{X}(S)$. We are going to define our reduced symplectic manifold as a $2n$-dimensional symplectic submanifold $(N,\omega|_N)$ of $S$ such that $X_H$ is nowhere tangent to $N$. We now need to induce a dynamics on $N$. A natural choice would be to consider the local foliation induced by letting evolve $N$ along the flow of $H$, however, this does not produce a proper dynamics on $N$ since it naturally identifies the motion with its initial value data, namely the integral curve with its point on $N$. Instead, we introduce a new vector field $T\in\mathfrak{X}(S)$ nowhere tangent to $N$ that will play the role of an external time. $T$ provides an integrable rank-1 distribution on $\T S$. By a Corollary of the Frobenius theorem \cite{Lee_2014}, for each $x\in N$ there exists a open neighborhood $U$ of $p$ and coordinates $(x_0,x_1,\dots,x_{2n})$ with $x_i\in(-\varepsilon,\varepsilon)$, in which $N\cap U$ is described by the equation $x_0=0$ and $T=\pder{}{x_0}$. Hence, we have a natural local diffeomorphism for each fixed $x_0=t$, namely $F_{t}:U_{x_0=t}\to N$ given by $F_{t}(x_0=t,x_1,\dots,x_{2n})=(0,x_1,\dots,x_{2n})$. Thus, we can project the Hamiltonian vector field $X_H$ for each fixed time obtaining a family of time-dependent vector field $X_t=dF_t(X_H)$ on $N$.\\
One can show that, under some general regularity condition for the Hamiltonian $H$, and a specific assumption for the symplectic form,
\begin{equation}
    \mathcal{L}_{X_t}\omega|_N=0.
\end{equation}
Thus, if $N$ is simply connected, there exists a family of time-dependent Hamiltonian $H_t$ on $N$ such that $X_{H_t}=X_t$.
\begin{proof}
    Suppose $H:\M\to \R$ is a constant rank map with ${\rm rk}(H)=1$. By the inverse function theorem, there exists a local set of coordinates $(y_0,y_1,\dots,y_{2n},x_{2n+1})$ for $\M$ such that the level set $S=H^{-1}(0)$ is locally described by $x_{2n+1}=f(y_0,\dots,y_{2n})$. For what we discussed before, there exists a local change of coordinates on $S$ such that $N$ is described by $x_0=0$. Composing this change of coordinates with the one given above, we obtain a new coordinate system $(x_0,x_1,\dots,x_{2n},x_{2n+1})$ with the properties just described. In this set of coordinates, the symplectic form on $\M$ reads:
    \begin{equation}
        \omega=\omega_adx^0\wedge dx^a+\omega_{ab}dx^a\wedge dx^b+\mathring{\omega}_adx^a\wedge dx^{2n+1}+\mathring{\omega}\,dx^0\wedge dx^{2n+1},
    \end{equation}
    where $a,b$ run from $1$ to $2n$. We know that the Hamiltonian vector field is tangent to $S$, hence $X_H|_S\inn\omega|_S=(X_H\inn\omega)|_S=(dH)|_S=0$. Moreover:
    \begin{equation}
        \begin{split}
            &\omega|_S=\left(\omega_a-\mathring{\omega}_a\pder{f}{x^0}+\mathring{\omega}\pder{f}{x^a}\right)\bigg|_S dx^0\wedge dx^a+\left(\omega_{ab}+\mathring{\omega}_a\pder{f}{x^b}\right)\bigg|_Sdx^a\wedge dx^b,\\
        & X_H|_S=X^a\pder{}{x^a}+X^0\pder{}{x^0},
        \end{split}
    \end{equation}
    where $X^0$ and $X^a$ depend on $(x_0,\dots,x_{2n})$ only. The equation $X_H\inn\omega=0$ provides a set of $2n+2$ equations on $S$ (by now all the functions are considered evaluated on $S$):
    \begin{equation}
        \begin{split}
            &\pder{H}{x^0}=-\frac{1}{2}\omega_aX^a,\\
        &\pder{H}{x^{2n+1}}=\frac{1}{2}(\mathring{\omega}X^0+\mathring{\omega}_aX^a),\\
        &\pder{H}{x^a}=\frac{1}{2}X^0\omega_a+\omega_{ba}X^b.
        \end{split}
    \end{equation}
    Furthermore, in this coordinates, the map $F_t:(x_0=t,x_1,\dots,x_{2n})\mapsto(0,x_1,\dots,x_{2n})$ provides vector fields
    \begin{equation}
        X_t=X^a|_{x_0=t}\pder{}{x^a}.
    \end{equation}
    The symplectic structure $\omega|_N$ has a simple form since it is just the second term of $\omega|_S$ evaluated on $x_0=0$. Hence
    \begin{equation}
        X_t\inn\omega|_N=X^a\omega_{ab}\big|_{x_0=0}dx^b+\frac{1}{2}\left(\mathring{\omega}_aX^b\pder{f}{x^b}-\mathring{\omega}_bX^b\pder{f}{x^a}\right)\bigg|_{x_0=0}dx^a.
    \end{equation}
    Whenever $\omega_a=0=\mathring{\omega}_a$, we get
    \begin{equation}
         X_t\inn\omega|_N=\pder{H}{x^a}\bigg|_{x_0=t}dx^a.
    \end{equation}
    which is a closed form on $N$. From the Cartan magic formula follows $\mathcal{L}_{X_t}\omega|_N=0$, hence each $X_t$ is a symplectic vector field. If $N$ is simply connected, namely $H_{\rm dR}^1(N)=0$, every closed form is exact, and a vector field is symplectic if and only if it is Hamiltonian.    
\end{proof}
From the proof is clear that our specific assumption of the symplectic form is that the flat coordinates satisfy $\{x^0,x^a\}=0=\{x^a,x^{2n+1}\}$.

\subsection{Cosmology as a constrained GUP theory}
An important example of this procedure is provided by Cosmology. We are going to consider Bianchi models in Misner variables in a GUP scenario \cite{ Misner:1969ae}. The phase space is $\M=\T^*\R^3$ and with coordinates $(q_0,q_1,q_2,p_0,p_1,p_2)$ satisfying the following non-vanishing Poisson brackets
\begin{equation}
    \begin{split}
        &\{q_i,p_j\}=f(p)\delta_{ij},\\
        &\{q_1,q_2\}=l(q,p),
    \end{split}
\end{equation}
and Hamiltonian given by
\begin{equation}
    H=-p_0^2+p_1^2+p_2^2+e^{4q_0}V(q_1,q_2),
\end{equation}
where $V$ is a positive smooth function which characterize the model. It is easy to find the symplectic form in this set of coordinates
\begin{equation}
     \omega_{ab}=\frac{1}{f}
    \begin{pmatrix}
    0 & 0 & 0 &-1 & 0 & 0\\
    0 & 0 & 0 & 0 &-1 & 0\\
    0 & 0 & 0 & 0 & 0 & -1\\
    1 & 0 & 0 & 0 & 0 & 0\\
    0 & 1 & 0 & 0 & 0 & \frac{l}{f}\\
    0 & 0 & 1 & 0 &-\frac{l}{f} & 0
   \end{pmatrix}.
\end{equation}
We recall that, by the closure condition \cite{Bruno_Segreto_2024}, here $l(q,p)$ must be necessarily:
\[l(q,p)=\pder{f}{p_1}q_2-\pder{f}{p_2}q_1.\]
We can now proceed to study the dynamics. The Hamiltonian vector field can be computed and it results to be:
\begin{equation}
    \begin{split}
        X_H=&2p_0f\pder{}{q_0}-\left(2p_1f+e^{4q_0}l\pder{V}{q_2}\right)\pder{}{q_1}-\left(2p_2f-e^{4q_0}l\pder{V}{q_1}\right)\pder{}{q_2}\\
        &+4e^{4q_0}fV\pder{}{p_0}+e^{4q_0}f\pder{V}{q_1}\pder{}{p_1}+e^{4q_0}f\pder{V}{q_2}\pder{}{p_2}.
    \end{split}
\end{equation}
The constraint surface is a fiber bundle over the configuration space, in which the fiber is a two-sheeted hyperboloid. Indeed, for each fixed $q$, the equation $p^2_0=p^2_1+p^2_2+e^{4q_0}V(q_1,q_2)$ describe a hyperboloid, and the condition $V>0$ ensure the ellipticity. Moreover, since the base manifold is contractible, the fiber bundle is trivial. Hence, we can choose $(q_0,q_1,q_2,p_1,p_2)$ as a set of global coordinates for the positive branch $S={\rm Graph}\left(\sqrt{p^2_1+p^2_2+e^{4q_0}V(q_1,q_2)}\right)$. The Hamiltonian vector field is tangent to $S$ and reads as:
\begin{equation}
\begin{split}
    X_H|_S=&2\sqrt{p^2_1+p^2_2+e^{4q_0}V(q_1,q_2)}f\pder{}{q_0}+e^{4q_0}f\pder{V}{q_1}\pder{}{p_1}+e^{4q_0}f\pder{V}{q_2}\pder{}{p_2}\\
        &-\left(2p_1f+e^{4q_0}l\pder{V}{q_2}\right)\pder{}{q_1}-\left(2p_2f-e^{4q_0}l\pder{V}{q_1}\right)\pder{}{q_2}.
\end{split}
\end{equation}
The symplectic manifold $N$ is given by the submanifold of $S$ identified by $q_0=0$. On this submanifold, the induced symplectic form is calculated to be:
\begin{equation}
    \omega|_N=\frac{1}{f(p)}\left(dp_1\wedge dq_1+dp_2\wedge dq_2+\frac{l(q,p)}{f(p)}dp_1\wedge dp_2\right).
\end{equation}
We choose the time vector field to be just $T=\pder{}{q_0}$, in such a way our coordinate system is a global flat chart for a tubular neighborhood on $N$, and so the family of maps $F_{t}(q_0=t,q_1,q_2,p_1,p_2)=(0,q_1,q_2,p_1,p_2)$ defines a family ${X_t}$ of vector fields on $N$, namely:
\begin{equation}
\begin{split}
    X_t=dF_t(X_H|_S)&=e^{4t}f\pder{V}{q_1}\pder{}{p_1}+e^{4t}f\pder{V}{q_2}\pder{}{p_2}-\left(2p_1f+e^{4t}l\pder{V}{q_2}\right)\pder{}{q_1}\\
    &-\left(2p_2f-e^{4t}l\pder{V}{q_1}\right)\pder{}{q_2}.
\end{split}
\end{equation}
Now, to compute $\mathcal{L}_{X_t}\omega|_N$, it is convenient to use the Cartan magic formula $\mathcal{L}_{X_t}\omega|_N=d(X_t\inn\omega|_N)$, here $d$ is the exterior derivative on $N$. Let us first compute the quantity $X_t\inn\omega|_N$:
\begin{equation}
    X_t\inn\omega|_N=e^{4t}\pder{V}{q_1}dq_1+e^{4t}\pder{V}{q_2}dq_2+2p_1dp_1+2p_2dp_2,
\end{equation}
which is clearly a closed 1-form. Since $N\cong\R^4$ is simply connected, any closed 1-form is also exact, hence $\mathcal{L}_{X_t}\omega|_N=0$. Finally, it is easy to find a family of functions $\{H_t\}$ such that $X_t=dH_t$, explicitly:
\begin{equation}
    H_t=p_1^2+p_2^2+e^{4t}V(p_1,p_2).
\end{equation}
which is the reduced Hamiltonian on $N$, the reduced phase space, whose geometrical structure is coherently dictated by $\omega|_N$.

\subsection{Conceptual comments on the Cosmology part}

Given a constrained Hamiltonian system \((\mathcal{M}, \omega, H)\), the general geometric procedure outlined in the previous sections allows us to successfully identify the reduced phase space of the theory, defined by the symplectic submanifold \(N\) along with the symplectic form \(\omega|_N\), which is the restriction of the general symplectic form \(\omega\) to the symplectic submanifold.  
The importance of this prescription, especially in the cosmological context, lies in the fact that it allows us to overcome an ambiguity in the construction of classical GUP theories for constrained Hamiltonian systems.  
Most of the studies scattered in the literature start from the imposition of a deformed symplectic form \emph{directly} on the submanifold defining the reduced phase space \cite{segreto2025dynamics, battisti2009mixmaster}.  
In the spirit of an effective approach to the dynamics, the rationale behind this procedure still holds.  
Indeed, the degrees of freedom surviving the reduction procedure can be considered the physical ones. As such, in an effective framework, they can be regarded as the only degrees of freedom that should be affected by a deformation of the dynamics.  
Despite this, this imposition by hand entails bypassing the geometric structure of the theory and clearly ignores the relation between the imposed symplectic form and the original symplectic form of the whole phase space, raising doubts about the physical and formal validity of the GUP method in this dynamical context.

The developed geometric framework settles the matter by providing a precise procedure of reduction and clarifying the connection between the two approaches.  
As is clearly shown in the given proof, given a suitable deformed symplectic form, concerning \emph{all} the degrees of freedom of the theory, we can obtain the symplectic form dictating the geometry of the reduced phase space as a proper restriction of the original symplectic form to the symplectic submanifold \(N\) on which the dynamics takes place.

From the discussion on the cosmological sector, it can be seen how the induced symplectic form  
has essentially the same structure as the original one, limited to the reduced degrees of freedom. This is exactly the general symplectic form usually imposed by hand on the reduced phase space in the effective method.  
This provides a rigorous motivation for the direct employment of the reduced symplectic form, showing clearly how it is connected with a symplectic form of a certain kind, involving all the degrees of freedom of the theory.  
As a consequence, the prescription we provide defines a general method and the necessary premise for the study of the dynamics of GUP-deformed constrained Hamiltonian systems.

In this respect, it is important to comment on the general structure of the symplectic form relative to the whole phase space.  
As shown clearly in the proof, we need, in the chosen coordinate system, to satisfy the following conditions:
\[
    \{x^0, x^a\} = 0 = \{x^a, x^{2n+1}\}.
\]
Given the role played by the \(x^0\) coordinate, which represents our time variable, these Poisson brackets set a requirement for commutativity between the space variables and the time variable.  
In other words, the deformation of the symplectic form which affects the space coordinates cannot be extended to the time coordinate.  
If this condition is violated, at least within our construction, the Lie derivative along the family of Hamiltonian fields of the induced symplectic form is different from zero.  
This equals the breakdown of the Hamiltonian formalism, since it implies that the dynamics does not conserve the geometry of the phase space.  
In particular, this signals a non-conservation of the phase space volume and violation of the reverse-time symmetry.  
On an intuitive level, this seems to be in agreement with the well-established fact that quantum non-commutative theories where non-commutativity is extended also to the time generator lead to non-unitary evolution (e.g., non-unitary S-matrix, causality problems in field theories \cite{gomis2000space, salminen2011noncommutative, chu2002hermitian, balachandran2004unitary}).  
Loosely speaking, the impossibility of defining a Hamiltonian formalism implies ill-definiteness of the Hamiltonian itself, which, when promoted to an operator in the quantum framework, could lead to non-unitary evolution due to the possible loss of self-adjointness or symmetry.  
Clearly, we cannot exclude that an extension of the present procedure might be possible, in order to recover the conservation of the induced symplectic form. Nevertheless, as the present procedure stands, this imposes a precise constraint on the form of the original \(\omega\) and excludes the possibility of non-commutativity among space and time variables.

\section{Conclusion} \label{Concl}
In this work, we have further developed the construction of a well-defined classical formulation for GUP theories, building on the analysis initiated in~\cite{Bruno_Segreto_2024}, where the general requirements for the existence of a symplectic structure were established. The focus here has been on understanding the behavior of the symplectic structure when constraints are present.
More precisely, given a GUP-deformed symplectic form $\omega$ on a phase space $\mathcal{M}$, we examined how this structure behaves under symplectic reduction or when restricted to a specific submanifold. Since constrained dynamical systems are ubiquitous in physics, particularly in gauge theories and gravitational models, this question naturally acquires a broad degree of generality.

\smallskip

We analyzed two main scenarios. The first corresponds to systems in which the constraints arise from a symmetry encoded in the action of a Lie group $G$. We focused on rotationally invariant GUP algebras in two and three dimensions. For these systems, we employed the standard symplectic reduction procedure, constructing the reduced phase space as the quotient $\mu^{-1}(0)/G$, where $\mu$ is the momentum map, identified here with the (deformed) angular momentum, and $\mu^{-1}(0)$ is the corresponding level set that defines the constraint surface. Once this quotient is formed, the symplectic two-form can be consistently projected to obtain its reduced counterpart, which governs the dynamics on the reduced phase space. As expected, the general structure of the symplectic two-form is preserved under this procedure, now restricted to the true physical degrees of freedom of the system. 

\smallskip

On the other hand, we considered the case in which the theory is characterized by a single constraint—namely, the Hamiltonian itself. In this situation, the symplectic quotient method cannot be applied, since, in general, no group action is available. Nevertheless, we provided a general prescription to identify an appropriate reduced phase space as a submanifold of the constraint surface. We successfully induced a dynamics (as a one-parameter family of Hamiltonians) on this submanifold by selecting an external time variable in the form of a vector field on the phase space, and we reconstructed the corresponding reduced symplectic form. Once again, the resulting structure preserves the formal properties of the original symplectic two-form $\omega$, restricted to the physical degrees of freedom. However, for this preservation to hold, some requirements must be imposed on the general GUP symplectic structure; in particular, noncommutativity between the prospective time variable and the spatial variables cannot occur. This constitutes a significant result, as it establishes necessary criteria for the physical viability of a GUP symplectic form—at least within the framework of our construction.\\
Given the promising role of GUP frameworks in mathematical cosmology, especially in the semiclassical analysis of cosmological models toward the initial singularity, we devoted special attention to understanding the dynamics emerging from the restriction of the symplectic form to the reduced phase space defined by the Hamiltonian constraint $H=0$. To this end, we worked out explicit computations for the general Hamiltonian of the Bianchi models in the Misner variables. The application of the machinery developed here leads us to obtain the reduced phase space, the resulting (time-dependent) Hamiltonian governing the dynamics, and the explicit expression of $\omega$ restricted to the submanifold, expressed in terms of the selected physical degrees of freedom. This final result is particularly noteworthy. Since the reduced symplectic form coincides with the one typically imposed \emph{ad hoc} on the reduced phase space — usually obtained by assuming the GUP two-form only after the reduction — it provides a rigorous justification for those dynamical treatments. In particular, it shows that the naive procedure employed in the literature is, in fact, validated by the correct geometric reduction, thereby offering the appropriate prescription for handling dynamics on the reduced phase space in the context of GUP theories.

\smallskip

In summary, the work furthers the construction of a consistent formulation of GUP theories, exploring the behavior of the GUP-symplectic structure whenever some kind of constraints defines the dynamics of the system of interest.\\
In this way, the prescription we gave provides a consistent and coherent way to study the GUP dynamics of a large class of physical systems, with particular attention to the mathematical cosmological sector, whose models are potentially relevant for the study of the dynamics of the very early universe \cite{ellis2006bianchi}.

\section*{Acknowledgement}
M.B. is supported by the MUR FIS2 Advanced Grant ET-NOW (CUP:~B53C25001080001) and by the INFN TEONGRAV initiative.

\appendix
\section{Vector bundles and invariant vector fields}
\label{App:A}
In this Appendix, we review a couple of useful facts about group actions on vector bundles. First of all we recall the definition of (smooth) vector bundle.
\begin{definition}
Let $M$ be a smooth manifold. A (real) \emph{vector bundle} of rank $r$ over $M$ is a smooth manifold $E$ together with a surjective smooth map $\pi : E \to M$ satisfying the following conditions:
\begin{itemize}
    \item For each $x \in M$, the fiber $E_x \doteq \pi^{-1}(x)$ is endowed with the structure of an $r$-dimensional real vector space,
    \item $E$ is locally trivial: for each $x \in M$, there exists a neighborhood $U$ of $x$ in $M$ and a diffeomorphism $\varphi : \pi^{-1}(U) \to U \times \mathbb{R}^r$ such that $\operatorname{pr}_U \circ \varphi = \pi$, and the restriction $\varphi|_{E_x}$ is a vector space isomorphism from $E_x$ to $\mathbb{R}^r$.
\end{itemize}
\end{definition}
A local \emph{vector field} is a smooth local section $\sigma : U \subset M \to E$ of the vector bundle, where $U$ is an open subset of $M$. A (global) vector field is a global smooth section $\sigma : M \to E$. A local \emph{frame} for $E$ over an open set $U \subset M$ is an ordered collection of $r$ local sections of $E$, $(\sigma_1, \dots, \sigma_r)$, such that for each $x \in U$, the $r$-tuple $(\sigma_1(x), \dots, \sigma_r(x))$ forms a basis of the fiber $E_x$. If $U = M$, it is called a \emph{global frame} or a \emph{frame of vector fields}.

Every local frame for a vector bundle is associated with a local trivialization. As a consequence, a vector bundle is trivial, i.e. it is diffeomorphic to $M\times \R^r$, if and only if it admits a global frame.

\medskip

Let us consider a group action $\alpha:G\times M\to M$ of a Lie group $G$ on $M$. Suppose there exists a lift of this action on the vector bundle
\begin{equation}
    \tilde\alpha:G\times E\to E \quad \text{such that} \quad \pi\circ\tilde\alpha=\alpha\circ\pi,
\end{equation}
and it sends fiber into fiber linearly, i.e. the map $\tilde\alpha_g:E_x\to E_{\alpha_g(x)}$ is linear for any $g\in G$. In such case, we say $E$ together with the action $\tilde\alpha$ is an invariant vector bundle. An \textit{invariant vector field} is a section such that
\begin{equation}
    \sigma\circ \alpha_g=\tilde\alpha_g\circ \sigma, \quad \text{for any }g\in G.
\end{equation}
Clearly, to be well-defined $\alpha_g(U)\subset U$ for all $g\in G$. An invariant frame is a frame composed by invariant vectors fields only.

\begin{proposition}
    An $\tilde\alpha$-invariant vector bundle is trivial, i.e. diffeomorphic to $M\times\R^r$ together with the action $\alpha\times\mathrm{id}$, if and only if it admits a global invariant frame.
\end{proposition}
\begin{proof}
    Suppose $E$ is trivial, namely there exists a diffeomorphism $\phi:E\to M\times \R^n$ such that $\alpha\times\mathrm{id}=\phi\circ \tilde\alpha\circ \phi^{-1}$ and linear on the fibers. Let $\{e_i\}$ be the canonical basis of $\R^r$, then the frame $\sigma_i$, with $1\leq i\leq r$, defined by 
    \begin{equation*}
        \sigma_i(x)=\phi^{-1}(x,e_i), \text{for all } x \in M
    \end{equation*}
    is a global invariant frame. Indeed, the vector fields are global by definition, moreover $\{\sigma_i(x)\}$ generates $E_x$ because every $u\in E_x$ can be written as the preimage of a unique $v\in \R^r$, and so as a linear combinations of $\sigma_i(x)$:
    \begin{equation}
        u=\phi^{-1}(x,v)=\phi^{-1}(x,v^ie_i)=v^i\phi^{-1}(x,e_i)=v^i\sigma_i(x).
    \end{equation}
    The set $\{\sigma_i(x)\}$ is also a set of linear independent vectors for every $x\in M$; indeed let us assume that there exists a vanishing linear combination of vectors $\sigma_i(x)$ with coefficients $a^i$ not all zero,
    \begin{equation}
        0=a^i\sigma_i(x)=\phi^{-1}(x,a^ie_i),
    \end{equation}
    then, by linearity of $\phi$ on the fibers, it implies that $a^ie_i=0$, contradicting the hypothesis that $e_i$ is a basis.\\
    Furthermore, the vector fields $\sigma_i$ are invariant:
\begin{equation*}
    \tilde \alpha_g \circ\sigma_i(x)=\tilde\alpha_g\circ \phi^{-1}(x,e_i)=\phi^{-1}(\alpha_g(x),e_i)=\sigma_i\circ \alpha_g(x), \quad \forall x\in M,\,\forall g \in G.
\end{equation*}

For the converse, let $\nu_i:M\to E$ be a global invariant frame. Let us consider the diffeomorphism given by
\begin{equation*}
\begin{matrix}
    \phi:&E&\to &M\times \R^r;\\
    &u&\mapsto & (\pi(u),u^i),
\end{matrix}
\end{equation*}
where $u^i$ are the components $u=u^i\nu_i(\pi(u))$. The inverse is immediate: $\phi^{-1}(x,u^i)=u^i\nu_i(x)$. It is easy to show that it is linear. Finally, it intertwines the actions:
\begin{equation*}
    \phi\circ \tilde \alpha_g(u)=(\alpha_g(\pi(u)),u'^i)=(\alpha_g(\pi(u)),u^i)=(\alpha_g\times\mathrm{id})\circ \phi(u).
\end{equation*}
Indeed, $u'^i=u^i$. In fact. $u'^i$ are such that $\tilde\alpha_g(u)=u'^i\nu_i(\alpha_g(\pi(u)))$, form which follows
\begin{equation*}
    u^i\tilde\alpha_g\circ \nu_i(\pi(u))=\tilde\alpha_g(u^i\nu_i(\pi(u)))= \tilde\alpha_g(u)=u'^i\tilde\alpha_g\circ \nu_i(\pi(u)),
\end{equation*}
and so the equivalence.
\end{proof}

\section{Cosmological 2d Yang-Mills theory}
\label{App:B}
The constrained Hamiltonian system introduced in Section~\ref{Sec:3-d} can be regarded as a two-dimensional cosmological $SO(3)$ gauge field theory. Consider a two-dimensional spacetime $M=\mathbb{R}\times S^1$ endowed with the flat Lorentzian metric
\[
ds^2_M=-dt^2+d\theta^2,
\]
where $t$ is the coordinate along $\mathbb{R}$, while $\theta\in(0,2\pi)$ is the standard angular coordinate on $S^1$.  

The cosmological assumption of homogeneity requires us to consider a transitive action of a Lie group on the spacelike foliation. In this setting, the natural action is $U(1)$ acting on $M$ as $(t,\theta)\mapsto(t,\theta+\varphi)$. Clearly, $d\theta$ is invariant under this action, and thus the dual vector $\partial_{\theta}$ can be regarded as the generator of the Lie algebra $\mathfrak{u}(1)$.  

By the cosmological assumption, the initial data must be given by a homogeneous $SO(3)$ gauge potential on $S^1$ at a fixed time $t_0$. These fields are classified by Wang’s theorem \cite{Wang_1958}, which establishes a one-to-one correspondence with linear maps $\phi:\mathfrak{u}(1)\to\mathfrak{so}(3)$. Consequently, every homogeneous gauge potential on $S^1$ can be written as
\begin{equation}
    A=\phi\circ\vartheta_{MC},
\end{equation}
where $\vartheta_{MC}$ is the Maurer--Cartan form on $U(1)$. From a geometric perspective, it is immediate to observe the existence of a residual global $SO(3)$ gauge symmetry. Moreover, the property of homogeneity must be preserved during time evolution. This approach is well established in the context of gauge theories \cite{Harnad_Shnider_Vinet_1980,Bojowald_Kastrup_2000}, although it has been subject to some criticism (for applications in cosmology, see \cite{Bojowald_2000,Bruno_23,Bruno_2025}).

Concretely, let us consider a generic $SO(3)$ gauge potential
\begin{equation}
    A(t,\theta)=A_0(t,\theta)\,dt+A_1(t,\theta)\,d\theta,
\end{equation}
where $A_0(t,\theta)$ and $A_1(t,\theta)$ are $\mathfrak{so}(3)$-valued functions. Homogeneity requires that these functions are independent of $\theta$. We may expand them in a basis $\{T_i\}$ of $\mathfrak{so}(3)$:
\[
T_1=\begin{pmatrix}
0 & 0 & 0\\
0 & 0 & 1\\
0 & -1 & 0
\end{pmatrix},\quad
T_2=\begin{pmatrix}
0 & 0 & -1\\
0 & 0 & 0\\
1 & 0 & 0
\end{pmatrix},\quad
T_3=\begin{pmatrix}
0 & 1 & 0\\
-1 & 0 & 0\\
0 & 0 & 0
\end{pmatrix}.
\]
Thus we obtain
\begin{equation}
    \begin{split}
        A_1(t)&=A^1_1(t)T_1+A^2_1(t)T_2+A^3_1(t)T_3\\
&\doteq q_1(t)T_1+q_2(t)T_2+q_3(t)T_3
=\begin{pmatrix}
0 & q_3 & -q_2\\
-q_3 & 0 & q_1\\
q_2 & -q_1 & 0
\end{pmatrix}.
    \end{split}
\end{equation}
The variables $q_1,\,q_2,\,q_3$ will serve as our configurational degrees of freedom.  

\medskip
To construct the constrained Hamiltonian system, let us start with a simple field theory: a two-dimensional $SO(3)$ \textrm{BF} theory. Although the Hamiltonian analysis of this theory is well known, we reproduce the computation here for the sake of self-consistency.  

In BF theory, one introduces an $\mathfrak{so}(3)$-valued function $B$, which can be decomposed as
\[
B = B^iT_i,
\]
with the homogeneity condition again enforcing independence from $\theta$. As we shall see, the $B^i$ play the role of conjugate momenta, $B^i\equiv p_i$. The action is given by
\begin{equation}
    S_{\rm BF}=\int_M \mathrm{tr}(BF),
\end{equation}
which expands to
\begin{equation}
    \begin{split}
        S_{\rm BF}&=\int_M \delta_{ij}B^iF^i_{\mu\nu}\,dx^{\mu}\wedge dx^{\nu}=\int_{\mathbb{R}\times S^1} B^iF^i_{01}\,dt\,d\theta\\
&=\int_{\mathbb{R}\times S^1} B^i\bigl(\partial_0A_1^i-\partial_1A_0^i+\epsilon_{ijk}A^j_0A^k_1\bigr)\,dt\,d\theta.
    \end{split}
\end{equation}
From this action we obtain the momenta
\begin{equation}
    \frac{\delta S_{\rm BF}}{\delta(\partial_0A_0^i)}=0,
    \qquad
    \frac{\delta S_{\rm BF}}{\delta(\partial_0A_1^i)}=B^i.
\end{equation}
Thus, the Hamiltonian contains three Lagrange multipliers $\lambda^i$ and three formal conjugate momenta $\Pi^i$ associated with $A_0^i$:
\begin{equation}
    \mathbf{H}=\int_{S^1} d\theta \,\bigl(B^i(\partial_0A_1^i)+\lambda^i\Pi_i-\mathcal{L}\bigr)
=\int_{S^1} d\theta \,\bigl(\lambda^i\Pi_i-A_0^i(\partial_1B^i+\epsilon_{ijk}A^j_1B^k)\bigr).
\end{equation}
The fundamental Poisson brackets on the phase space are
\begin{equation}
\begin{split}
\{A^i_0(t,\theta),\Pi^j(t,\theta')\}&=\delta^{ij}\delta(\theta-\theta'),\\
\{A^i_1(t,\theta),B^j(t,\theta')\}&=\delta^{ij}\delta(\theta-\theta').
\end{split}
\end{equation}
We thus obtain the primary constraint $\Pi^i=0$. Ensuring consistency under time evolution generates the secondary constraint
\begin{equation}
    \{\Pi^i,\mathbf{H}\}=\partial_1B^i+\epsilon_{ijk}A^j_1B^k \equiv 0.
\end{equation}
Therefore, $A_0^i$ has trivial dynamics and acts solely as a Lagrange multiplier. 

\smallskip

Starting from the action $S_{\rm BF}$ and imposing homogeneity, the integration over $S^1$ contributes only an overall constant factor. We then found the homogeneous BF action
\begin{equation}
    S_{\rm BF}^{\rm hom}=\int_{\mathbb{R}} p_i\bigl(\dot q^i+\epsilon_{ijk}\Lambda^jq^k\bigr)\,dt,
\end{equation}
where $B^i=p_i$, $A^i_1=q^i$, and $\Lambda^i=A^i_0$. We are thus left with a final-dimensional phase space with fundamental Poisson brackets
\begin{equation}
    \begin{split}
        \{\Lambda^i,\Pi_j\}=\delta^i_j,\\
\{q^i,p_j\}=\delta^{i}_j;
    \end{split}
\end{equation}
and a fully constrained Hamiltonian
\begin{equation}
    \mathbf{H}=\lambda^i\Pi_i+\Lambda^i\epsilon_{ijk}q^jp^k.
\end{equation}
Indeed, by the request of trivial time evolution of the constraint $\Pi_i=0$, we obtain a secondary constraint
\begin{equation}
    \{\Pi^i,\mathbf{H}\}=\epsilon_{ijk}q^jp^k\equiv 0
\end{equation}
As expected, since BF theory is a topological field theory, the Hamiltonian is entirely constrained. Moreover, it coincides with the Gauss constraint, which generates gauge transformations, in agreement with the fact that the dynamical evolution in a topological field theory reduces to gauge transformations.

\medskip

We now turn to a less trivial example, namely a two-dimensional Yang--Mills theory with gauge group $SO(3)$ over $M$. The action is
\begin{equation}
    S_{\rm YM}=-\frac{1}{4}\int_M \mathrm{tr}(F\wedge \star F).
\end{equation}
In two dimensions, $\star F$ reduces to an $\mathfrak{so}(3)$-valued function that can be written as
\[
\star F=-F_{01}^iT_i, 
\qquad 
F_{01}^i=\partial_0A_1^i-\partial_1A_0^i+\epsilon_{ijk}A^j_0A^k_1.
\]
Hence the action becomes
\begin{equation}
    S_{\rm YM}=\frac{1}{2}\int_{\mathbb{R}\times S^1}\sum_{i=1}^3 (F^i_{01})^2\,dt\,d\theta.
\end{equation}
From this action we obtain the canonical momenta:
\begin{equation}
    \frac{\delta S_{\rm YM}}{\delta(\partial_0A_0^i)}=0,
\qquad
\frac{\delta S_{\rm YM}}{\delta(\partial_0A_1^i)}=F^i_{01}\equiv E^i.
\end{equation}
The Hamiltonian is then computed as
\begin{equation}
    \begin{split}
        \mathbf{H}&=\int_{S^1} d\theta \,\bigl(E^i(\partial_0A_1^i)+\lambda^i\Pi_i-\mathcal{L}\bigr) \\
&=\int_{S^1} d\theta \,\bigl(\lambda^i\Pi_i+E^i(\partial_1A_0^i-\epsilon_{ijk}A^j_0A^k_1)+\tfrac{1}{2}E^iE^i\bigr).
    \end{split}
\end{equation}
As before, we find a primary constraint $\Pi^i\equiv 0$, whose consistency under time evolution gives rise to the Gauss constraint:
\begin{equation}
    \{\Pi^i,\mathbf{H}\}=\partial_1E^i+\epsilon_{ijk}A^j_1E^k \equiv 0.
\end{equation}

\smallskip

As before, we impose homogeneity and find the homogeneous Yang-Mills action
\begin{equation}
    S_{\rm YM}^{\rm hom}=\frac{1}{2}\int_{\mathbb{R}}\sum_{i=1}^3 (\dot q^i+\epsilon_{ijk}\Lambda^jq^k)^2\,dt.
\end{equation}
From which the canonical momenta are
\begin{equation}
    \Pi_i=\frac{\delta S_{\rm YM}^{\rm hom}}{\delta\dot\Lambda^i}=0,
\qquad
\frac{\delta S_{\rm YM}}{\delta\dot q^i}=\dot q^i+\epsilon_{ijk}\Lambda^jq^k\equiv p_i.
\end{equation}
Thus the Hamiltonian is
\begin{equation}
    \mathbf{H}=\lambda^i\Pi_i+\Lambda^i\epsilon_{ijk}q^jp^k+\frac{1}{2}p^ip^i.
\end{equation}
In this case, on the phase space $(q^i,p_i)$, the Gauss constraint is still present $\epsilon_{ijk}q^jp^k\equiv0$, generating gauge transformations. However, in contrast with the BF theory, there is also a genuine Hamiltonian function
\[
H=\tfrac{1}{2}p^ip^i,
\]
which is gauge invariant and descends to a smooth function on the symplectic quotient of the phase space, where the momentum map is given by the Gauss constraint.

\bibliographystyle{unsrt}
\bibliography{biblio}
\end{document}